\documentclass[11pt,reqno]{amsart}
\usepackage{amsfonts}
\usepackage{amsmath}
\usepackage{amsthm}
\usepackage{graphicx}

\usepackage{amssymb}
\usepackage{bm}
\usepackage{hyperref}
\usepackage{float}
\usepackage{mathtools}
\usepackage{enumerate}
\usepackage{tikz}
\usepackage{geometry}
\usepackage{xcolor}
\usetikzlibrary{matrix,arrows,positioning,automata}
\newtheorem{theorem}{Theorem}[section]

\newtheorem{corollary}[theorem]{Corollary}
\theoremstyle{definition}

\newtheorem{definition}[theorem]{Definition}

\def\E{{\mathbb E}}

\def\R{{\mathbb R}}

\def\PP{{\mathbb P}}
\def\FF{{\mathbb F}}

\def\A{{\mathcal A}}

\def\F{{\mathcal F}}

\def\val2{w}
\def\ZZ{{\mathcal Z}}
\def\newC{\lambda}
\def\newomega{\beta}

\title[Consumption and investment under relative performance criteria]{Many-player games of optimal consumption and investment under relative performance criteria}
\author{Daniel Lacker}
\address{Industrial Engineering \& Operations Research, Columbia University, New York, NY.}
\email{daniel.lacker@columbia.edu}
\author{Agathe Soret}
\address{Industrial Engineering \& Operations Research, Columbia University, New York, NY.}
\email{acs2298@columbia.edu}
\date{}

\begin{document}

\begin{abstract}
We study a portfolio optimization problem for competitive agents with CRRA utilities and a common finite time horizon. The utility of an agent depends not only on her absolute wealth and consumption but also on her relative wealth and consumption when compared to the averages among the other agents. We derive a closed form solution for the $n$-player game and the corresponding mean field game. This solution is unique in the class of equilibria with constant investment and continuous time-dependent consumption, both independent of the wealth of the agent. Compared to the classical Merton problem with one agent, the competitive model exhibits a wide range of highly nonlinear and non-monotone dependence on the agents' risk tolerance and competitiveness parameters. Counter-intuitively, competitive agents with high risk tolerance may behave like non-competitive agents with low risk tolerance.
\end{abstract}

\maketitle

\section{Introduction}

In this paper, we extend the CRRA model for the optimal investment problem recently developed by Lacker and Zariphopoulou \cite{lacker2017mean} to include consumption. 
Our model can be loosely described as follows, with full details given in Section \ref{finite population}. Each agent chooses consumption and investment policies, with access to a riskless bond and a lognormal stock. The stocks in which the different agents specialize can be correlated, and we cover the extreme case of perfect correlation, i.e., a single stock in which all agents trade. Each agent has a CRRA utility depending on both absolute and relative wealth at the (common) time horizon $T$ as well as absolute and relative consumption, the latter in a time-integrated sense. Agents have different levels of risk aversion and different preferences toward absolute versus relative performance. The relative performance criteria couple these $n$ optimization problems, and we find the (unique, in a sense to be clarified later) Nash equilibrium in terms of the various model parameters.
As is natural in light of the classical Merton problem \cite{merton1969lifetime}, in equilibrium each agent invests a constant fraction of wealth in the stock, and the consumption strategy is time-dependent but independent of the agent's wealth.

The equilibrium behavior fits with Samuelson's result \cite{samuelson1975lifetime}: the investment strategy is independent of the consumption strategy. That is, the investment strategy is exactly the same as in the model without consumption studied in \cite{lacker2017mean}.
The equilibrium consumption policy, as a function of the various model parameters, displays even more highly nonlinear and non-monotone behavior than the investment policy, and we study this in detail in Section \ref{se:discussion}. 
Notably, each agent's rate of consumption $c_t$ changes monotonically with time $t$ over the entire horizon $[0,T]$; however, whether an agent increases or decreases consumption over time depends in a complex manner on her own risk preferences as well as certain aggregates of the other agents' parameters.

Three key features of our model are relative consumption concerns, relative wealth concerns, and asset specialization. We defer to the introduction of \cite{lacker2017mean} for a thorough discussion of the latter two topics and further references, but we stress the particularly important and by now well-established point that mutual fund choice is highly influenced by relative performance \cite{sirri1998costly}. That is, out-performing other fund managers tends to attract greater future investment in one's own fund. The most closely related works to ours, after \cite{lacker2017mean}, are
\cite{anthropelos-geng-zariphopoulou,basak2015competition,
bielagk2017equilibrium,dosreis-zariphopoulou,espinosa2015optimal,frei2011financial}. These papers study continuous-time models of optimal investment under relative performance concerns in various settings, including different kinds of utilities, equilibrium pricing, and state constraints, but none incorporate consumption.

There are two natural arguments for studying relative consumption concerns.
On the one hand, interpreting agents as fund managers, we may think of consumption as capital accumulation, in the form of equipment, technology, or benefits for employees. A high relative consumption in this sense would naturally attract clientele or better employees, and more generally it should lead to similar benefits as a high relative wealth.
On the other hand, if we interpret the agents in our model as household investors, then relative consumption concerns fit naturally with models of \emph{keeping up with the Joneses}; this line of literature directly incorporates the social aspects of investment and consumption decisions \cite{abel1990asset,chan2002catching,demarzo2004diversification,gali1994keeping}.

Our paper contributes to the literatures on optimal consumption and investment as well as the application of mean field games. The dynamic problem of lifetime consumption and investment planning with one player was formalized and studied in the landmark papers of Merton \cite{merton1969lifetime,merton1975optimum} and Samuelson \cite{samuelson1975lifetime}. Later work incorporated more complex features into the models, such as general price processes, bankruptcy, etc.\ \cite{karatzas1987optimal,karatzas2003optimal}. Those that incorporated \emph{multiple agents} into the model, such as \cite{karatzas1997explicit,sethi1992infinite}, did so in an equilibrium context; each agent's behavior depends on the others only through the price, which is determined in equilibrium. Agents are price-takers in our model, and we do not attempt to incorporate price equilibrium, as this would severely strain tractability.

In another direction, our work provides a new explicitly solvable mean field game model. Mean field games, introduced in \cite{huang2006large,lasry-lions}, are rarely explicitly solvable outside of linear-quadratic examples. See \cite{bayraktar2018large,carmona2018dyson,gueant-lasry-lions,lacker2017mean,sun2016systemic} for some notable exceptions and the book \cite{carmona-delarue-book} for further background on the active area of mean field games.
From a mean field game perspective, our model is rather complex: It involves common noise, degenerate volatility coefficients, singular objective functions, and a mean field interaction through both the states and controls (i.e., an \emph{extended mean field game} \cite[Chapter I.4.6]{carmona-delarue-book}). Nevertheless, the precise structure of the problem lends itself to an explicit solution. Our argument follows along the lines of \cite{lacker2017mean}, treating the mean field term (geometric mean of wealth) as a state variable, which leads to a fixed point problem involving a single Hamilton-Jacobi-Bellman (HJB) equation as opposed to the $n$-dimensional HJB system often used for stochastic differential games. After showing that this equation admits a unique and separable classical solution, the fixed point is resolved via a system of non-linear ordinary differential equations. Despite many similarities with \cite{lacker2017mean}, the consumption renders the arguments substantially more involved.

The paper is organized as follows. In Section \ref{finite population}, we formulate and solve the $n$-agent model described above. Then, in Section $3$, we study the infinite population counterpart of this problem, arguing that the $n\to\infty$ limit results in a simpler form of the equilibrium. Finally, Section $4$ discusses and interprets the form of the equilibrium and its dependence on the model parameters.

\section{The n-agent game}\label{finite population}

In this section, we consider the $n$-player game, where each agent trades in a common investment horizon $[0,T]$. Agents may invest in their own specific stocks or in a common riskless bond which offers zero interest rate.
The price of stock $i$, in which only agent $i$ trades, is given by the dynamics
\begin{equation}\label{prices}
    \frac{dS^i_t}{S^i_t} = \mu_i dt + \nu_i dW^i_t + \sigma_i dB_t,
\end{equation}
where the Brownian motions $W^1, ..., W^n, B$ are independent and defined on a probability space $(\Omega, \F, \PP)$, which we endow with the natural filtration $(\F_t)_{t \in [0,T]}$ generated by these $n+1$ Brownian motions, and where the market parameters are constants $\mu_i > 0$, $\sigma_i \ge 0$, and $\nu_i \ge 0$, with $\sigma_i + \nu_i > 0$. The prices $S^i_t$ are assumed to be one-dimensional for simplicity, although we could easily extend the results to $k$-dimensional prices.

This setup covers the important special case of a \textit{single stock}, corresponding to the situation where all the stocks are identical. That is, $\mu_i = \mu$, $\nu_i = 0$, and $\sigma_i = \sigma$, for all $i = 1, ..., n$ and for some $\mu, \sigma > 0$ independent of $i$, and so $S^i \equiv S^j$ for each $i,j$ (assuming the initial values agree). In the single stock case, the agents face the same market opportunities rather than specializing in different assets, though their risk preferences still differ.

Each agent $i$ chooses a self-financing strategy, $(\pi^i_t)_{t \in [0,T]}$, denoting the proportion of wealth invested in the stock $i$, and a consumption policy, $(c_t^i)_{t \in [0,T]}$. The wealth process of agent $i$ is then given by
\begin{equation}\label{processX}
    dX^i_t = \pi^i_t X_t^i \left(\mu_i dt + \nu_i dW^i_t + \sigma_i dB_t\right) - c^i_t X^i_t dt, \quad\quad X_0^i = x_0^i.
\end{equation}
Note that $c^i_tX^i_t$ represents the instantaneous rate of consumption of agent $i$, so that $c^i_t$ is the rate per unit wealth. We say that a portfolio strategy is admissible if it belongs to the set $\A$ of $\FF$-progressively measurable $\R \times \R_+$-valued process $(\pi_t, c_t)_{t \in [0,T]}$ satisfying $\E  \int_0^T (\pi_t^2 + c_t^2)dt < \infty$. Throughout the paper, $\R_+ :=(0,\infty)$ denotes the \emph{strictly} positive reals, and we do not allow a consumption rate of zero. This is reasonable and less restrictive than it may at first appear because the form of our utility functions, introduced in the next paragraph, will ensure that agents' marginal utilities approach $+\infty$ as consumption approaches zero. Note also that for any admissible portfolio strategy we have $X^i_t > 0$ for all $t \in [0,T]$. Indeed, we parametrized our consumption process as we did in \eqref{processX} in part to avoid the possibility of bankruptcy and in part to avoid imposing any state constraints.

The utility function of each agent belongs to the family of power (CRRA) utilities,
\begin{equation*}
U(x; \delta) = \begin{cases}
        \frac{1}{1 - 1 / \delta}x^{1 - 1 / \delta} & \text{ if } \delta \ne 1 \\
        \log x & \text{ if } \delta = 1,
    \end{cases}
\end{equation*}
defined for $x,\delta > 0$. Agent $i$ seeks to maximize the expected utility
\begin{align}
J_i((\pi^i,c^i)_{i=1}^n) = \E\left[\int_0^T   U\left(c^i_tX^i_t  (\overline{cX}_t)^{- \theta_i}; \delta_i \right)dt + \epsilon_i U\left(X^i_T \overline{X}_T^{- \theta_i}; \delta_i \right)\right], \label{def:utility}
\end{align}
defined for any vector of admissible strategies $(\pi^i,c^i)_{i=1}^n$ where $(\pi^i,c^i) \in \A$ for each $i=1,\ldots,n$.
Here $\overline{X}_T = \left( \prod_{k = 1}^{n} X_T^k \right)^{1/n}$ and $\overline{cX}_t = \left( \prod_{k = 1}^{n} (c_tX_t)^k \right)^{1/n}$ are the population (geometric) average wealth and consumption rate, respectively. 
The parameters $\delta_i > 0$  and $\theta_i \in [0,1]$ represent respectively the $i^{\mathrm{th}}$ agent's risk tolerance and competition weight. We apply the same utility function to both wealth and consumption for tractability reasons, but we scale the utility of wealth with the parameter $\epsilon_i > 0$ to capture the relative importance that the agent assigns to terminal wealth compared to consumption.

Our choice to work with geometric averages instead of arithmetic averages is also motivated by tractability: Just as the (arithmetic) average of independent Brownian motions is again a Brownian motion, the geometric average of independent geometric Brownian motions is again a geometric Brownian motion. It is worth emphasizing another form of the utility function, revealed by writing the terms inside the utility function as
\begin{align*}
c^i_tX^i_t  (\overline{cX}_t)^{- \theta_i} = (c^i_tX^i_t)^{1-\theta_i}\left(\frac{c^i_tX_t^i}{\overline{cX}_t}\right)^{\theta_i}, \quad\quad X^i_T \overline{X}_T^{- \theta_i} = (X^i_T)^{1-\theta_i}\left(\frac{X_T^i}{\overline{X}_T}\right)^{\theta_i}.
\end{align*}
The ratios $c^i_tX_t^i/\overline{cX}_t$ and $X^i_T/\overline{X}_T$ measure the relative consumption rate and relative terminal wealth, respectively. In particular, the utility function in \eqref{def:utility} is applied to the log-convex combination between absolute and relative consumption rate and terminal wealth, with $\theta_i$ controlling the tradeoff between absolute and relative performance. For $\theta_i$ close to $1$, agent $i$ is more concerned with relative performance than absolute performance, and for $\theta_i=0$ agent $i$ is not at all competitive and ignores the rest of the population.

The goal is to find a Nash equilibrium, an investment strategy $(\vec{\pi}^*_t, \vec{c}^*_t)_{t \in [0,T]}$ such that $\pi_t^{i,*}$ and  $c_t^{i,*}$ are respectively the optimal stock and consumption allocation exercised by agent $i$ in response to the strategy of all the other agents. With Merton's problem and the recent findings of \cite{lacker2017mean} in mind, we might expect to find an equilibrium where the investment strategies $\vec{\pi}^*_t$ are constant and the consumption strategies $c_t^{i,*}$ are only time-dependent.

\begin{definition}
We say that a vector $(\pi^{i,*}, c^{i,*})_{i=1}^n$ of admissible strategies (i.e., $(\pi^{i,^*}, c^{i,*}) \in \A$ for each $i$) is an \emph{equilibrium} if for each $i = 1, ..., n$ and each $(\pi, c) \in \A$ we have 
\begin{equation*}
 J_i((\pi^{i,*},c^{i,*})_{i=1}^n) \ge J_i\Big( \ldots, (\pi^{i-1,*},c^{i-1,*}), (\pi, c), (\pi^{i+1,*},c^{i+1,*}), \ldots\Big).
\end{equation*}
An equilibrium $(\pi^{i,*}, c^{i,*})_{i=1}^n$ is called a \emph{strong equilibrium} if, for each $i$, the process $c^{i,*}$ is deterministic and continuous, and the process $\pi^{i,*}$ is deterministic and constant.\footnote{This definition of equilibrium is more specifically of \emph{open-loop} type, but a strong equilibrium, being nonrandom, can be shown to also provide an equilibrium over closed-loop or Markovian controls.}
\end{definition}

The main result is the following, which gives the explicit form of an equilibrium:

\begin{theorem}\label{thmN}
Let $n \ge 2$.
Assume that for all $i = 1, ..., n$, we have $x_0^i > 0$, $\delta_i > 0$, $\theta_i \in [0,T]$, $\epsilon_i > 0$, $\mu_i > 0$, $\sigma_i \ge 0$, $\nu_i \ge 0$, and $\sigma_i + \nu_i > 0$. Then there is a unique strong equilibrium $(\pi^{i,*}, c^{i,*})_{i=1}^n$, and it takes the following form:
\begin{align}
 \pi^{i,*} &= \frac{\delta_i \mu_i }{\sigma_i^2 + \nu^2_i(1 + (\delta_i -1)\theta_i/n)} -  \frac{\theta_i (\delta_i -1)\sigma_i}{\sigma_i^2 + \nu_i^2(1 + (\delta_i -1)\theta_i/n)} \frac{\phi}{1 + \psi} \label{optcontroleq} \\
c^{i,*}_t &= \begin{cases}
            \left(\frac{1}{\newomega_i} + \left(\frac{1}{\newC_i} - \frac{1}{\newomega_i} \right)e^{- \newomega_i (T - t)} \right)^{-1} & \mathrm{if } \  \beta_i  \neq 0 \\
            (T - t + \newC_i^{-1})^{-1} & \mathrm{if } \ \beta_i = 0. 
        \end{cases} \label{optconsumption}
\end{align}
The constants $(\phi,\psi)$ and $(\newomega_i,\newC_i)_{i=1}^n$ are given by
\begin{align}
\phi &= \frac{1}{n} \sum_{k = 1}^n \delta_k \frac{\mu_k \sigma_k}{\sigma_k^2 + \nu_k^2(1 + (\delta_k-1)\theta_k/n)}, \nonumber \\
\psi &= \frac{1}{n} \sum_{k = 1}^n  \theta_k (\delta_k - 1)\frac{\sigma_k^2}{\sigma_k^2 + \nu_k^2(1 + (\delta_k-1)\theta_k/n)}, \nonumber \\
\newomega_i &= \theta_i (\delta_i - 1)\frac{\frac{1}{n} \sum_{k = 1}^n \delta_k \rho_k}{1 + \frac{1}{n} \sum_{k = 1}^n \theta_k(\delta_k -1)  }- \delta_i \rho_i, \nonumber \\
\newC_i &= \epsilon_i^{-\delta_i} \left(\left(\prod_{k = 1}^n \epsilon_k^{\delta_k} \right)^{1/n}\right)^{\theta_i(\delta_i - 1) / (1 + \frac{1}{n} \sum_{k = 1}^n \theta_k(\delta_k-1))}, \label{def:C_i}
\end{align}
where we define also $(\rho_i)_{i=1}^n$ by
\begin{align*}
\rho_i = (1 - 1/\delta_i) \Bigg\{ & \frac{(1 - \theta_i/n)(\mu_i - \sigma_i\theta_i(1 - 1/\delta_i)\frac{1}{n}\sum_{k \ne i} \sigma_k \pi^{k,*})^2}{2 (\sigma_i^2 + \nu_i^2)(1-(1 -\theta_i/n)(1 - 1/\delta_i) )}  \\
  &+ \frac{1}{2}\Big(\Big(\frac{1}{n} \sum_{k\ne i} \sigma_k \pi^{k,*} \Big)^2 + \frac{1}{n^2} \sum_{k \ne i} (\nu_k \pi^{k,*})^2\Big) \theta_i^2(1 - 1/\delta_i) \\
  &-  \theta_i \frac{1}{n} \sum_{k \ne i} \mu_k \pi^{k,*} + \frac{\theta_i}{2n} \sum_{k \ne i} ( \sigma^2_k + \nu_k^2)(\pi^{k,*})^2 \Bigg\},
\end{align*}
Moreover, we have the identity
\begin{align}
\frac{1}{n} \sum_{k=1}^{n} \sigma_k \pi^{k,*} = \frac{\phi}{1 + \psi}. \label{def:vol-identity}
\end{align}
\end{theorem}

Note that $\delta_i=1$ implies $\beta_i=0$, which means that log-investors always use the second form of $c^{i,*}_t$ given in \eqref{optconsumption}.
The form of the equilibrium does not seem to simplify much further, except in the single stock case:

\begin{corollary}(Single stock)
Assume that for all $i = 1, ..., n$ we have $\mu_i = \mu > 0$, $\sigma_i = \sigma > 0$, and $\nu_i =0$. Then there is a unique strong equilibrium $(\pi^{i,*}, c^{i,*})_{i=1}^n$, and it takes the following form:
\begin{align}
\pi^{i,*} &= \frac{\mu}{\sigma^2}\left(\delta_i - \frac{\theta_i}{\theta_{\mathrm{crit}}}(\delta_i-1)\right) \label{optcontroleq-singlestock} \\
c^{i,*}_t &= \begin{cases}
            \left(\frac{1}{\newomega_i} + \left(\frac{1}{\newC_i} - \frac{1}{\newomega_i} \right)e^{- \newomega_i (T - t)} \right)^{-1} & \mathrm{if } \ \beta_i \ne 0, \\
            (T - t + \newC_i^{-1})^{-1} & \mathrm{if } \ \beta_i = 0. 
        \end{cases} \label{optconsumption-singlestock}
\end{align}
For each $i$, the constant $\newC_i$ is given by \eqref{def:C_i}, and $\newomega_i$ and $\theta_{\mathrm{crit}}$ are given by
\begin{align*}
\newomega_i &= \frac{\mu^2}{2\sigma^2}(1 - \delta_i)\left(1 - \frac{\theta_i}{\theta_{\mathrm{crit}}}\right)\left(\delta_i - \frac{\theta_i}{\theta_{\mathrm{crit}}}(\delta_i-1)\right), \\
\theta_{\mathrm{crit}} &= \frac{1 + \frac{1}{n}\sum_{k = 1}^{n} \theta_k(\delta_k -1)}{\frac{1}{n}\sum_{k = 1}^{n} \delta_k }.
\end{align*}
\end{corollary}

In Section \ref{se:mfg} we simplify this further by sending $n\to\infty$. Then, in Section \ref{se:discussion}, we analyze in detail how the equilibrium behavior depends on the various parameters.

\begin{proof}[Proof of Theorem \ref{thmN}]
First, we fix an agent $i \in \{1,\ldots,n\}$ and suppose that all other players follow given strategies. That is, for $k \neq i$, let $\pi_k \in \R$ and $c_k : [0,T] \to \R_+$ denote fixed admissible strategies for the other agents, in which the investment policy $\pi_k$ is constant and the consumption policy $c_k$ is a deterministic continuous function.
We will solve the optimization problem for agent $i$, determining the agent's best response to the competitors' strategies. Then, we will resolved the resulting fixed point problem.

Define $Y_t := (\prod_{k \ne i} X_t^k)^{1/n}$, where $X_t^k$ solves (\ref{processX}) subject to the strategies $(\pi_k,c_k)$, with $X_0^k = x^k_0$. We use the following abbreviations:
\[
\begin{array}{ccc}
\Sigma_k = \sigma_k^2 + \nu_k^2 & \quad \widehat{\mu \pi}_{-i} = \frac{1}{n} \sum_{k \ne i} \mu_k \pi_k, & \quad  \widehat{\sigma \pi}_{-i} = \frac{1}{n} \sum_{k \ne i} \sigma_k \pi_k, \\
\widehat{\Sigma \pi^2}_{-i} = \frac{1}{n} \sum_{k \ne i} \Sigma_k \pi_k^2, & \quad \widehat{(\nu \pi)^2}_{-i} = \frac{1}{n} \sum_{k \ne i} \nu_k^2 \pi_k^2, & \quad \widehat{c}_{-i}(t) = \frac{1}{n} \sum_{k \ne i} c_k(t). \\
\end{array}
\]
A straightforward calculation with It\^o's formula (cf.\ the proof of Theorem 14 in \cite{lacker2017mean}) shows that the process $Y_t$ satisfies
\begin{equation}\label{processY}
\frac{dY_t}{Y_t} = (\eta_i - \widehat{c}_{-i}(t)) dt + \frac{1}{n} \sum_{k \ne i} \nu_k \pi_k dW^k_t + \widehat{\sigma \pi}_{-i} dB_t,
\end{equation}
where we define also
\[
\eta_i = \widehat{\mu \pi}_{-i} - \frac{1}{2}\left( \widehat{\Sigma \pi^2}_{-i} - \widehat{\sigma \pi}_{-i}^2 - \frac{1}{n} \widehat{(\nu \pi)^2}_{-i}\right).
\]
The $i^{\mathrm{th}}$ agent then solves the optimization problem 
\begin{equation}\label{opt}
\sup_{(\pi^i, c^i) \in \A} \E\left[\int_0^T  U\left((c^i_t X^i_t)^{1 - \theta_i /n} (\bar{c}_{-i}(t)Y_t)^{- \theta_i}; \delta_i \right)dt + \epsilon_i U\left(\left(X^i_T\right)^{1 - \theta_i / n} Y_T^{- \theta_i}; \delta_i \right) \right],
\end{equation}
where $\bar{c}_{-i}(t) = \left( \prod_{k \ne i} c_k(t) \right)^{1/n}$ and\footnote{We use a bar $\overline{c}$ to denote a geometric average and a hat $\widehat{c}$ to denote an arithmetic average.}
\[
dX^i_t = \pi^i_t X_t^i (\mu_i dt + \nu_i dW^i_t + \sigma_i dB_t) - c^i_t X^i_t dt, \hspace{0.5 cm} X_0^i = x_0^i,
\]
with $(Y_t)_{t \in [0,T]}$ solving (\ref{processY}). Treating $(X^i,Y)$ as the state process, we solve this stochastic optimal control problem by noting that the value (\ref{opt}) should equal $v(X_0^i, Y_0, 0)$, where $v(x,y,t)$ solves the HJB equation
\begin{equation}
\begin{split}
0=v_t &+ \sup_{\pi \in \R} \left[ \pi(\mu_i x v_x + \sigma_i \widehat{\sigma \pi}_{-i} xy v_{xy}) + \frac{1}{2}  \pi^2\Sigma_i x^2 v_{xx}\right] \\
&+ \sup_{c \in \R_+} \left[- c x  v_x + U\left((cx)^{(1 - \theta_i/n)}(\bar{c}_{-i}(t)y)^{- \theta_i}; \delta_i \right) \right] \\
&+ (\eta_i - \widehat{c}_{-i}(t))y v_y + \frac{1}{2}\left(\frac{1}{n} \widehat{(\nu \pi)^2}_{-i} + \widehat{\sigma \pi}_{-i}^2 \right) y^2 v_{yy},
\end{split} \label{pf:HJB}
\end{equation}
for $(x,y,t) \in \R_+ \times \R_+ \times [0,T)$, with terminal condition
\begin{align}
v(x,y,T) = \epsilon_i  U(x^{1 - \theta_i/n}y^{- \theta_i}; \delta_i). \label{pf:HJB-boundary}
\end{align}
Notice that the two suprema in \eqref{pf:HJB} are finite if $v_{xx} < 0$ and $v_x > 0$, so we assume for the moment that this is the case, and we will ultimately check that our solution does satisfy these constraints.
If follows from a standard verification theorem that there can be at most one classical solution of this PDE.
Since the utility takes a different form depending on whether or not $\delta_i = 1$, we treat these two cases separately in the next part of the proof. \\ {\ }

\noindent\textbf{The case} $\bm{\delta_i \neq 1}$:
The utility function takes the form 
    \begin{equation*}
           U\left((cx)^{\left(1 - \theta_i/n\right)}(\bar{c}_{-i}(t)y)^{-\theta_i}; \delta_i \right) = \left(1 - \frac{1}{\delta_i} \right)^{-1} (cx)^{(1 - \theta_i/n)(1 - 1 / \delta_i)}(\bar{c}_{-i}(t)y)^{- \theta_i(1- 1 / \delta_i)}.
    \end{equation*}
Applying the first order conditions, the suprema in \eqref{pf:HJB} are attained by
\begin{equation}\label{optcontrolNot}
    \pi^{i,*}(x,y,t) = - \frac{\mu_i x v_x(x,y,t) + \sigma_i \widehat{\sigma \pi}_{-i} xy v_{xy}(x,y,t)}{\Sigma_i x^2 v_{xx}(x,y,t)},
\end{equation}
and
\begin{equation}
c^{i,*}(x,y,t) = \frac{1}{x}\left(\frac{(1 - \frac{\theta_i}{n})(\bar{c}_{-i}(t)y)^{- \theta_i(1- 1 / \delta_i)}}{ v_x(x,y,t)} \right)^{\frac{1}{1-(1 - \theta_i/n)(1- 1/\delta_i)}}.
    \tag{\theequation'}
\end{equation}
Let us introduce the following constants:
\begin{align}\label{gammai}
\gamma_i &= \frac{1}{1-(1 - \theta_i/n)(1- 1/\delta_i)}, \\
\Gamma_i &= \left(1 - \frac{\theta_i}{n}\right)^{ \gamma_i}\left(\frac{1}{ \gamma_i - 1} \right). \nonumber 
\end{align}
Using these constants and the expressions \eqref{optcontrolNot} and (\ref{optcontrolNot}'), the HJB equation \eqref{pf:HJB} becomes 
\begin{equation}\label{HJBnot}
\begin{split}
0 = v_t &- \frac{1}{2}\frac{(\mu_i x v_x + \sigma_i \widehat{\sigma \pi}_{-i} xy v_{xy})^2}{\Sigma_i x^2 v_{xx}} + \frac{1}{2}\big(\widehat{\sigma \pi}_{-i}^2 + \frac{1}{n}\widehat{(\nu \pi)^2}_{-i}\big)y^2 v_{yy} + (\eta_i - \widehat{c}_{-i}(t)) y v_y \\
    &+ (v_x)^{1 - \gamma_i}(\bar{c}_{-i}(t)y)^{- \gamma_i \theta_i(1- 1 / \delta_i)}\Gamma_i.
    \end{split}
\end{equation}
We now make the ansatz 
\begin{align}
v(x,y,t) = \epsilon_i \left( 1 - \frac{1}{\delta_i} \right)^{-1}  x^{(1 - \theta_i/n)(1 - 1 / \delta_i)} y^{-\theta_i(1 - 1 / \delta_i)}f_i(t), \label{ansatz-Not}
\end{align}
for a differentiable function $f_i : [0,T] \to \R$ to be determined.
Note that the boundary condition \eqref{pf:HJB-boundary} requires $f_i(T)=1$. 
Plugging this into the HJB equation \eqref{HJBnot}, we find that $v(x,y,t)/f_i(t)$ factors out of each term, and we get
\begin{equation} \label{eqf}
0 = f_i'(t) + \left(\rho_i +  \theta_i \left(1 - \frac{1}{\delta_i}\right)\widehat{c}_{-i}(t)\right) f_i(t) + \frac{\epsilon_i^{-\gamma_i}}{\gamma_i} \bar{c}_{-i}(t)^{- \gamma_i \theta_i(1- 1 / \delta_i)} f_i(t)^{1 - \gamma_i},
\end{equation}
where we define
\begin{equation}\label{rhoi}
\begin{split}
      \rho_i = \left(1 - \frac{1}{\delta_i}\right) & \left( \frac{\gamma_i (1 - \theta_i/n)(\mu_i - \sigma_i \widehat{\sigma \pi}_{-i}\theta_i (1 - 1 / \delta_i))^2}{2\Sigma_i} \right. \\
      &\quad + \left. \frac{1}{2}(\widehat{\sigma \pi}_{-i}^2 + \frac{1}{n}\widehat{(\nu \pi)^2}_{-i})\theta_i^2(1 - 1 / \delta_i) - \theta_i \widehat{\mu \pi}_{-i} + \frac{\theta_i}{2}\widehat{\Sigma \pi^2}_{-i}  \right).
\end{split}
\end{equation}
Indeed, the last term in \eqref{eqf} comes from the identity
\begin{equation*}
\begin{split}
v_x(x,y,t)^{1 - \gamma_i}y^{- \gamma_i \theta_i(1- 1 / \delta_i)} = \epsilon_i^{-\gamma_i} \left( 1 - \frac{\theta_i}{n} \right)^{1 - \gamma_i} \left( 1 - \frac{1}{\delta_i}\right) f_i(t)^{- \gamma_i} v(x,y,t).
\end{split}
\end{equation*}
To solve \eqref{eqf}, let us for the moment abbreviate
\begin{align}
a_i(t) &:= \rho_i + \theta_i(1 - 1/\delta_i)\widehat{c}_{-i}(t), \quad\qquad
b_i(t) := \frac{\epsilon_i^{-\gamma_i}}{\gamma_i} \bar{c}_{-i}(t)^{- \gamma_i \theta_i(1- 1 / \delta_i)}. \label{def:a-b-for-f}
\end{align}
Then (\ref{eqf}) rewrites as
\begin{equation*}
f_i' + a_i f_i + b_i f_i^{1 - \gamma_i} = 0.
\end{equation*}
This is an example of a \emph{Bernoulli equation}, and a well known change of variables leads to the solution.
Indeed, and divide by $f_i^{1-\gamma_i}$ (after noting that $\gamma_i > 0$) and use the substitution $u_i(t) = f_i^{\gamma_i}(t)$, so that (\ref{eqf}) becomes the linear differential equation
\begin{equation*}
\frac{1}{\gamma_i} u_i' + a_i u_i + b_i = 0,
\end{equation*}
with terminal condition $u_i(T) = 1$. This linear equation admits the unique solution 
\begin{equation*}
u_i(t) = e^{\gamma_i \int_t^T a_i(s)ds} + \int_t^T \gamma_i b_i(s) e^{- \gamma_i \int_t^s a_i(r)dr}ds.
\end{equation*}
Note that $b_i$ is positive everywhere, and thus so is $u_i$. Hence, $u_i^{1/\gamma_i}$ is well defined, and the unique solution to (\ref{eqf}) is given by
\begin{equation}\label{solBern}
f_i(t) = \left(e^{\gamma_i \int_t^T a_i(s)ds} + \int_t^T \gamma_i b_i(s) e^{- \gamma_i \int_t^s a_i(r)dr}ds \right)^{1/\gamma_i}.
\end{equation}
Substituting this solution \eqref{solBern} into the ansatz \eqref{ansatz-Not} yields the solution $v(x,y,t)$ of the HJB equation, as long as we check that $v_{xx} < 0$ and $v_x > 0$. But this is straightforward:
\begin{align*}
v_x(x,y,t) &= \epsilon_i (1 - \theta_i/n) x^{- 1 / \gamma_i} y^{- \theta_i( 1 - 1/\delta_i)} f_i(t) > 0, \\ 
v_{xx}(x,y,t) &= - \frac{\epsilon_i}{\gamma_i} (1 - \theta_i/n) x^{- 1 / \gamma_i - 1} y^{- \theta_i( 1 - 1/\delta_i)} f_i(t) < 0,
\end{align*}
where we again use $\gamma_i > 0$.
Therefore, in terms of $f_i$, we may express the optimal controls from \eqref{optcontrolNot} and (\ref{optcontrolNot}') as
\begin{equation}\label{optcontrolNotOne}
\begin{split}
\pi^{i,*} &=  \frac{\gamma_i(\mu_i - \sigma_i \widehat{\sigma \pi}_{-i}\theta_i (1 - 1 / \delta_i))}{\Sigma_i}, \\
c^{i,*}_t &= \epsilon_i^{-\gamma_i} (\bar{c}_{-i}(t))^{- \gamma_i \theta_i(1- 1 / \delta_i)} f_i^{- \gamma_i}(t).
\end{split}
\end{equation}
{ \ } \\

\noindent\textbf{The case} $\bm{\delta_i=1}$: 
In the case $\delta_i = 1$ we must proceed differently, but we will ultimately derive optimal controls that are consistent with the formulas in \eqref{optcontrolNotOne}. Note first that we may greatly simplify the form of \eqref{opt}, because the logarithmic utility function implies in particular that the other players no longer influence player $i$'s optimization. That is, player $i$ maximizes the simplified objective
\begin{equation}\label{valuefOne}
(1 - \theta_i /n) \E \left[ \int_0^T  \log (c^i_t X^i_t) dt + \epsilon_i  \log X^i_T \right].
\end{equation}
Noting that $1-\theta_i/n > 0$, the value (\ref{valuefOne}) is equal to $(1-\theta_i/n)\val2(X_0^i,0)$, where $\val2(x,t)$ solves the HJB equation
\begin{equation}
\begin{split}
0 = \val2_t &+ \max_{\pi \in \R} \left[ \pi \mu_i x \val2_x + \pi^2 \frac{1}{2} \Sigma_i x^2 \val2_{xx}\right] + \max_{c \in \R_+} \left[- c x  \val2_x + \log(cx) \right],
\end{split} \label{HJB-One-simplified}
\end{equation}
for $(x,t) \in \R_+ \times [0,T)$, with terminal condition $\val2(x,T) = \epsilon_i \log x$.
The maximum is attained by
\begin{equation}\label{optconsone}
\begin{split}
    \pi^{i,*}(x,t) &= - \frac{\mu_i x \val2_x(x,t)}{\Sigma_i x^2 \val2_{xx}(x,t)}, \quad\qquad
    c^{i,*}(x,t) = \frac{1}{x \val2_x(x,t)}.
\end{split}
\end{equation}
The HJB equation \eqref{HJB-One-simplified} then becomes 
\begin{equation}\label{HJBone}
0 = \val2_t - \frac{1}{2}\frac{(\mu_i x \val2_x)^2}{\Sigma_i x^2 \val2_{xx}} - 1 - \log  \val2_x.
\end{equation}
Make the ansatz
\[
\val2(x,t) = f_i(t)\epsilon_i\log x + g_i(t),
\]
where $f_i$ and $g_i$ are to be determined and satisfy $f_i(T)=1$ and $g_i(T)=0$.
Plug this into (\ref{HJBone}), defining the constant $\widehat\rho_i = \mu_i^2\epsilon_i / 2\Sigma_i$, to get
\[
\left( \epsilon_i f_i'(t) + 1 \right) \log x + g_i'(t) + \widehat\rho_i f_i(t) - 1 - \log \epsilon_i - \log f_i(t) = 0.
\]
Since there is only one term depending on $x$, we must have
\[
\epsilon_i f_i'(t) + 1 = 0, \quad\quad f_i(T)=1.
\]
This yields $f_i(t) = \epsilon_i^{-1}(T - t) +1$.
Then $g_i$ must solve
\[
g'_i(t)  =  - \widehat\rho_i (\epsilon_i^{-1}(T - t) +1) + 1 + \log ((T - t) + \epsilon_i),
\]
which is easily integrated using $g_i(T)=0$ to get the solution $g_i$, though we will not need to use it explicitly.
Note also that $f_i > 0$, and thus $w_x > 0$ and $w_{xx} < 0$ everywhere.
We have therefore solved the HJB when $\delta_i = 1$. Recalling (\ref{optconsone}), we deduce that the optimal controls are
\begin{align}\label{optcontrolOne}
\pi^{i,*} &= \frac{\mu_i}{\Sigma_i}, \quad\qquad c^{i,*}_t = \frac{1}{T - t + \epsilon_i}.
\end{align}
Note that the result \eqref{optcontrolNotOne} obtained above specializes to \eqref{optcontrolOne} when $\delta_i=1$. Indeed, if $\delta_i = 1$, then $\rho_i = 0$ and $\gamma_i = 1$, and the functions $a_i$ and $b_i$ defined in \eqref{def:a-b-for-f} reduce to $a_i \equiv 0$ and $b_i \equiv 1/\epsilon_i$. Thus \eqref{solBern} becomes $f_i(t) = \epsilon_i^{-1}(T- t)+1$, and \eqref{optcontrolNotOne} becomes \eqref{optcontrolOne}.

{ \ }

\noindent\textbf{Completing the proof:}
We now complete the proof, using the form we found above for the optimal control of player $i$ in response to the other players' choices. Namely, the best response of player $i$ is given by the controls in \eqref{optcontrolNotOne}, where $f_i$ is defined as in \eqref{solBern}, and we have seen that these formulas are valid for both cases $\delta_i=1$ and $\delta_i \neq 1$.

Note that if we assume that the other consumption functions are positive and continuous, then the optimal feedback consumption that we have found is also positive and continuous on $[0,T]$. Now, to conclude the proof, note that the original choice of $(\pi_i, c_i)_{i=1}^n$ is a strong equilibrium if and only if for each $i = 1, ..., n$, we have
\[
\pi^{i,*} = \pi_i \hspace{0,5cm} \text{ and } \hspace{0,5cm} c^{i,*}_t = c_i(t) \hspace{0,5cm} \forall t \in [0,T],
\]
where $(\pi^{i,*},c^{i,*})$ were given in \eqref{optcontrolNotOne}.

We first address the investment policy. Note that we obtained exactly the same optimal control $\pi^{i,*}$ as in the problem without consumption, and we can conclude as in the proof of \cite[Theorem 14]{lacker2017mean} that $\pi^{i,*} = \pi_i$ for all $i = 1, ..., n$ if and only if $\pi^{i,*}$ is as in \eqref{optcontroleq}. In addition, we prove the identity \eqref{def:vol-identity} just as in \cite[Theorem 14]{lacker2017mean}. Recall that $\rho_i$ defined in (\ref{rhoi}) depends on the investment policies $\pi_k$ of the other agents but not on the consumption policies; in particular, in equilibrium we have
\begin{equation}
\begin{split}
  \rho_i := \Big(1 - \frac{1}{\delta_i}\Big) &\Bigg\{ \frac{(1 - \theta_i/n)(\mu_i - \sigma_i\theta_i(1 - 1/\delta_i)\frac{1}{n}\sum_{k \ne i} \sigma_k \pi^{k,*})^2}{2 (\sigma_i^2 + \nu_i^2)(1 - (1 -\theta_i/n)(1 - 1/\delta_i))} \\
  &\quad+ \frac{1}{2}\theta_i^2\Big(1 - \frac{1}{\delta_i}\Big)\Bigg(\Big(\frac{1}{n} \sum_{k\ne i} \sigma_k \pi^{k,*} \Big)^2 + \frac{1}{n} \sum_{k \ne i} (\nu_k \pi^{k,*})^2\Bigg)  \\
  &\quad -   \theta_i \frac{1}{n} \sum_{k \ne i} \mu_k \pi^{k,*} + \frac{\theta_i}{2n} \sum_{k \ne i} ( \sigma^2_k + \nu_k^2)(\pi^{k,*})^2 \Bigg\}
  \end{split} \label{pf:rhoi-def}
\end{equation}

Next, we address the consumption policies.
In light of our arguments above, 
in order to have an equilibrium, we must simultaneously solve the following system of equations, for $i=1,\ldots,n$:
\begin{align}
c_i(t) &=  \epsilon_i^{-\gamma_i} (\bar{c}_{-i}(t))^{- \gamma_i \theta_i(1- 1 / \delta_i)} (f_i(t))^{- \gamma_i} \label{systeq1} \\
0 &= f_i'(t) + (\rho_i +  \theta_i (1 - 1 / \delta_i)\widehat{c}_{-i}(t)) f_i(t) + \frac{\epsilon_i^{-\gamma_i}}{\gamma_i} \bar{c}_{-i}(t)^{- \gamma_i \theta_i(1- 1 / \delta_i)} f_i(t)^{1 - \gamma_i}, \label{systeq2}
\end{align}
with $f_i(T) = 1$. Indeed, the first equation gives the best response of agent $i$ in terms of the other agents' strategies (computed in \eqref{optcontrolNotOne}) and the function $f_i$. The second equation is exactly the differential equation which determined $f_i$, which we solved explicitly in terms of the other agents' strategies in \eqref{solBern}. However, now that we have verified the validity of the ansatz for $v_i(x,y,t)$, to resolve the equilibrium it is more convenient to abandon the explicit form for $f_i$ and instead solve the equations \eqref{systeq1} and \eqref{systeq2} simultaneously.
To do this, first plug  \eqref{systeq1} into the last term of \eqref{systeq2} to find
\[
f_i'(t) + \left(\rho_i +  \theta_i \left(1 - \frac{1}{\delta_i}\right)\widehat{c}_{-i}(t) + \frac{1}{\gamma_i} c_i(t)\right) f_i(t) = 0.
\]
Defining the full average $\widehat{c}(t) = \frac{1}{n} \sum_{k = 1}^n c_k(t)$, note that $\widehat{c}_{-i}(t) = \widehat{c}(t) - c_i(t)/n$. Recalling the definition of $\gamma_i$ in \eqref{gammai}, 
we deduce that
\[
f_i'(t) + \left(\rho_i +  \theta_i \left(1 - \frac{1}{\delta_i}\right)\widehat{c}(t) + \frac{1}{\delta_i} c_i(t)\right) f_i(t) = 0.
\]
Hence, with $f_i(T) = 1$ this leads to
\begin{align}
f_i(t) = \exp \left( \int_t^T \left(\rho_i +  \theta_i \Big(1 - \frac{1}{\delta_i}\Big)\widehat{c}(s) + \frac{1}{\delta_i} c_i(s)\right) ds \right). \label{pf:f-expression}
\end{align}
Now notice that \eqref{systeq1} is equivalent to 
\begin{equation*}
c_i(t)^{1- \gamma_i(\theta_i/n)(1 - 1/\delta_i)} = \epsilon_i^{-\gamma_i} \bar{c}(t)^{- \gamma_i \theta_i(1 - 1/\delta_i)} f_i(t)^{-\gamma_i},
\end{equation*}
where $\bar{c}(t)$ denotes the full geometric average, $\bar{c}(t) := \left(\prod_{k = 1}^n c_k(t) \right)^{1/n}$. Hence, recalling the definition of $\gamma_i$ in \eqref{gammai},
\begin{equation*}
c_i(t) = \left(\epsilon_i f_i(t)\right)^{- \frac{\gamma_i}{1 - \gamma_i(\theta_i/n)(1 - 1/\delta_i)}}\bar{c}(t)^{- \frac{ \gamma_i \theta_i(1 - 1/\delta_i)}{1 - \gamma_i(\theta_i/n)(1 - 1/\delta_i)}}  = \left(\epsilon_i f_i(t)\right)^{-\delta_i}\bar{c}(t)^{- \theta_i (\delta_i - 1)} .
\end{equation*}
Next, plug in the expression for $f_i$ from \eqref{pf:f-expression} to get
\[
c_i(t) = \epsilon_i^{-\delta_i} \bar{c}(t)^{- \theta_i (\delta_i - 1)} \exp \left( - \delta_i \int_t^T \left(\rho_i +  \theta_i \Big(1 - \frac{1}{\delta_i}\Big)\widehat{c}(s) + \frac{1}{\delta_i} c_i(s)\right) ds \right),
\]
which is equivalent to
\begin{equation}\label{ci0}
c_i(t) \exp \left( \int_t^T  c_i(s) ds \right) = \epsilon_i^{-\delta_i} \bar{c}(t)^{- \theta_i (\delta_i - 1)} e^{- \delta_i \rho_i (T-t)} \exp \left(- \theta_i (\delta_i - 1) \int_t^T \widehat{c}(s) ds \right).
\end{equation}
Taking the geometric mean over $i=1,\ldots,n$, we get
\[
\bar{c}(t) \exp \left( \int_t^T  \widehat{c}(s) ds \right) = \left(\overline{\epsilon^{\delta}}\right)^{-1} \bar{c}(t)^{- \widehat{\theta(\delta -1)}}  e^{- \widehat{\delta \rho} (T-t)}\exp \left(- \widehat{\theta (\delta - 1)}\int_t^T \widehat{c}(s) ds \right),
\]
where we defined
\[
\begin{array}{lll}
\widehat{\theta (\delta - 1)} := \frac{1}{n} \sum_{k = 1}^n \theta_k (\delta_k -1),  &  \widehat{\delta \rho} := \frac{1}{n} \sum_{k = 1}^n \delta_k \rho_k, \text{ and} & \overline{\epsilon^{\delta}} := \left(\prod_{k = 1}^n \epsilon_k^{\delta_k} \right)^{1/n}. \\
\end{array}
\]
Thus,
\[
\bar{c}(t) \exp \left( \int_t^T  \widehat{c}(s) ds \right) = \left(\overline{\epsilon^{\delta}}\right)^{-\frac{1}{1 +\widehat{\theta(\delta-1)}}} e^{- \frac{\widehat{\delta \rho}}{1 + \widehat{\theta(\delta -1)}  } (T - t)}.
\]
Plugging this expression into \eqref{ci0}, we get
\begin{equation}\label{derivci}
c_i(t) \exp \left( \int_t^T  c_i(s) ds \right) =  \newC_i e^{ \newomega_i (T - t)},
\end{equation}
where we define
\begin{align*}
\newomega_i &:= \theta_i (\delta_i - 1)\frac{\widehat{\delta \rho}}{1 + \widehat{\theta(\delta -1)}  }- \delta_i \rho_i, \\
\newC_i &:= \epsilon_i^{- \delta_i} \left(\overline{\epsilon^{\delta}}\right)^{\frac{\theta_i(\delta_i -1 )}{1 +\widehat{\theta(\delta-1)}}} > 0.
\end{align*}
Integrate (\ref{derivci}) from $t$ to $T$ and take the logarithm to get
\begin{equation}
\int_t^T  c_i(s) ds = \begin{cases}
\log \left( 1 + \frac{\newC_i}{ \newomega_i} \left( e^{ \newomega_i (T - t)} - 1 \right) \right) &\text{if } \beta_i\neq 0 \\
 \log\left(\newC_i(T-t) + 1\right) &\text{if }  \beta_i = 0.
\end{cases}	\label{intci}
\end{equation}
This is indeed well defined because, when $\beta_i \neq 0$, the function $t \mapsto 1 + \frac{\newC_i}{ \newomega_i} \left( e^{ \newomega_i (T - t)} - 1 \right)$ is decreasing on $[0,T]$ and equal to $1$ at $t=T$. Differentiating both sides, we finally obtain
\[
c_i(t) = \begin{cases}
\left(\frac{1}{\newomega_i} + \left(\frac{1}{\newC_i} - \frac{1}{\newomega_i} \right)e^{- \newomega_i (T - t)} \right)^{-1} &\text{if } \beta_i\neq 0 \\
(T - t + \newC_i^{-1})^{-1} &\text{if } \beta_i = 0.
\end{cases}
\]
In summary, we have found the unique solution of the system of equations given in \eqref{systeq1} and \eqref{systeq2}, justifying our ansatz for the HJB equation \eqref{pf:HJB}. With a classical solution of the HJB equation in hand, by a standard verification argument \cite{fleming2006controlled,pham2009continuous} we conclude that the portfolio and consumption policies identified above do indeed provide the unique best responses and thus the unique strong equilibrium.
\end{proof}

\section{The mean field game} \label{se:mfg}

We study in this section the limit as $n \rightarrow \infty$ of the $n$-player game analyzed previously, and we explain how the limit can be viewed as the equilibrium outcome of a (mean field) game with a continuum of agents. For each agent $i$, define the \emph{type vector}
\[
\zeta_i := (x_0^i, \delta_i, \theta_i, \epsilon_i, \mu_i, \nu_i, \sigma_i).
\]
We now allow these parameters to depend also on $n$, though we will not burden the notation with an additional index.
These type vectors induce an empirical measure, the {\it type distribution}, which is the probability measure on the {\it type space}
\[
\ZZ := (0, \infty) \times (0, \infty) \times [0,1] \times (0, \infty) \times (0, \infty) \times [0, \infty) \times [0, \infty),
\]
given by $m_n = \frac{1}{n}\sum_{k=1}^n\delta_{\zeta_k}$.
Now note that for each agent $i$, the equilibrium strategy for the consumption as well as the investment only depends on the agent's own type vector and on the distribution $m_n$ of the type vectors. Hence, if we assume $m_n$ converges weakly to some limiting probability measure, then we expect the equilibrium outcome to converge in a certain sense.

In order to pass to the limit, let us now denote by $(x_0,\delta,\theta,\epsilon,\mu,\nu,\sigma)$ a $\ZZ$-valued random variable, with $\nu+\sigma > 0$ a.s.
This law of this random type vector represents the distribution of type vectors of a continuum of agents, and a single realization of this random type vector is to be interpreted as the type assigned to a single representative agent.
We assume that all expectations appearing in this paragraph are finite.
The $n\to\infty$ limiting forms of the constants $\phi$ and $\psi$ defined in Theorem \ref{thmN} are as follows:
\begin{align}
\phi &= \E\left[\frac{\delta \mu \sigma}{\sigma^2 + \nu^2} \right], \qquad \psi = \E\left[\frac{\theta (\delta -1) \sigma^2}{\sigma^2 + \nu^2} \right].  \label{MF-phi}
\end{align}
To identify limiting forms of the remaining quantities in Theorem \ref{thmN}, we additionally remove the $i$ subscript, letting the randomness of the type vector play the role of the names of the agents. This gives
\begin{align}
\newomega &= \theta(\delta -1) \frac{\E \left[\delta \rho \right]}{1 + \E \left[ \theta(\delta -1)\right]} - \delta \rho, \label{MF-omega}
\end{align}
and 
\begin{equation}
\begin{split}
   \rho = \left(1 - \frac{1}{\delta}\right) 
       &\left\{  \frac{\delta}{2(\sigma^2+\nu^2)} \left(\mu - \sigma \frac{\phi}{1 + \psi}\theta (1 - 1 / \delta)\right)^2 + \frac{1}{2}\left(\frac{\phi}{1 + \psi} \right)^2 \theta^2(1 - 1 / \delta)\right.\\
       &-  \theta\frac{\phi}{1 + \psi} \E\left[\frac{\delta \mu^2 - \theta (\delta -1)\sigma \mu}{\sigma^2 + \nu^2} \right] + \left. \frac{\theta}{2}\E \left[\frac{(\delta \mu - \theta (\delta -1)\sigma \frac{\phi}{1+\psi})^2}{\sigma^2 + \nu^2} \right]  \right\}.  \label{MF-rho}
\end{split}
\end{equation}
The limiting form of $\newC_i$ is given by
\begin{equation}
\newC = \epsilon^{-\delta} \left( e^{\E\left[\log(\epsilon^{-\delta})\right]}\right)^{-\frac{\theta(\delta-1)}{1 + \E[\theta(\delta-1)]}}. \label{MF-C}
\end{equation}
Indeed, this is determined by noting that
\[
\left(\prod_{k=1}^n\epsilon_k^{\delta_k}\right)^{1/n} = \exp\left(\frac{1}{n}\sum_{k=1}^n \log(\epsilon_k^{\delta_k})\right).
\]
The equilibrium investment policy of Theorem \ref{thmN} of the representative agent then becomes
\begin{align}
\pi^* = \frac{\delta \mu}{\sigma^2 + \nu^2} - \frac{\theta(\delta -1) \sigma}{\sigma^2 + \nu^2}\frac{\phi}{1 + \psi} \label{MF-pi*}
\end{align}
and the consumption policy becomes 
\begin{align} \label{MF-c*}
c^*_t = \begin{cases}
    \left(\frac{1}{\newomega} + \left(\frac{1}{\newC} - \frac{1}{\newomega} \right)e^{- \newomega (T - t)} \right)^{-1} & \text{ if }  \beta \neq 0 \\
    (T - t + \newC^{-1})^{-1} & \text{ if } \beta=0.
\end{cases}
\end{align}

We next illustrate how this strategy arises as the equilibrium of a mean field game. 
Let $(\Omega,\F,\FF=(\F_t)_{t \in [0,T]},\PP)$ be a filtered probability space supporting independent Brownian motions $B$ and $W$ as well as a random type vector $\zeta= (\xi, \delta, \theta, \epsilon, \mu, \nu, \sigma)$ as above. Assume that $\FF$ is the minimal complete filtration with respect to which $\zeta$ is $\F_0$-measurable and $W$ and $B$ are $\FF$-Brownian motions.
The representative agent's wealth process is determined by
\begin{equation}\label{wealth}
dX_t = \pi_t X_t (\mu dt + \nu dW_t + \sigma dB_t) - c_t X_t dt.
\end{equation}
As before, admissible strategies are given by $\FF$-progressively measurable $\R \times \R_+$-valued processes $(\pi,c)$ satisfying $\E\int_0^T(\pi_t^2 + c_t^2)dt < \infty$, and every admissible strategy results in a strictly positive wealth process.

Because this is a mean field game with common noise $B$, the mean field equilibrium condition will involve \emph{conditional} means given $B$. Intuitively, because the interaction between the agents occurs through the (geometric) average over the whole population, we expect some kind of a law of large numbers and asymptotic independence between the agents as $n\to\infty$. Due to the presence of common noise, any asymptotic independence between the agents must be conditional on the common noise $B$, and we refer to \cite{carmona-delarue-book,carmona-delarue-lacker} for more thorough and precise treatments of mean field games with common noise. In other words, the population average wealth and consumption processes should be adapted to the complete filtration $\FF^B=(\F^B_t)_{t \in [0,T]}$ generated by the common noise $B$.
Now, suppose that the representative agent knows that the geometric mean wealth and consumption of the (continuum of) other agents are governed by some $\FF^B$-adapted processes $\overline{X}$ and $\overline{\Gamma}$, respectively.
Then, the objective of the representative agent is to maximize the expected payoff
\begin{equation}\label{optpb}
\sup_{(\pi, c) \in \mathcal{A}_{MF}} \E \left[\int_0^T U\left(c_t X_t (\overline{\Gamma}_t \overline{X}_t)^{-\theta}; \delta \right)dt + \epsilon U \left(X_T \overline{X}^{-\theta}_T; \delta \right)\right].
\end{equation}
In equilibrium, the optimal $(\pi^*,c^*)$ for this problem should lead to $\exp\E[\log X_t \, | \, \F^B_t] = \overline{X}_t$ and $\exp\E[\log c^*_t \, | \, \F^B_t] = \overline{\Gamma}_t$, where we note  that $\exp\E[\log(\cdot)]$ is the continuous analogue of geometric mean.
We formalize this discussion in the following definition:

\begin{definition} Let $(\pi^*, c^*)$ be admissible strategies, and consider the $\FF^{B}$-adapted processes $\overline{X}_t := \exp \E[\log X^*_t | \mathcal{F}_{t}^B]$ and $\overline{\Gamma}_t = \exp \E[ \log c^*_t | \mathcal{F}_t^B]$, where $(X_t^*)_{t \in [0,T]}$ is the wealth process in (\ref{wealth}) corresponding to the strategy $(\pi^*, c^*)$. We say that $(\pi^*, c^*)$ is a \emph{mean field equilibrium} if $(\pi^*, c^*)$ is optimal for the optimization problem (\ref{optpb}) corresponding to this choice of $\overline{X}$ and  $\overline{\Gamma}$. We call $(\pi^*, c^*)$ a \emph{strong equilibrium} if $\pi^*$ is constant (and thus $\F_0$-measurable) and if $(c^*_t)_{t \in [0,T]}$ is continuous and $\F_0$-measurable.
\end{definition}

Because $\F_0$ is the $\sigma$-field generated by the type vector, to say that a strategy is $\F_0$-measurable simply means that it depends on the type vector only, not on the Brownian motions or wealth process.
We may now state a theorem which explains the precise sense in which the $n\to\infty$ limiting strategies computed above can be viewed as the equilibrium outcome of a mean field game.
 
\begin{theorem}\label{thmMF}
Assume that a.s.\ $\delta > 0$, $\theta \in [0,1]$, $\epsilon > 0$, $\mu > 0$, $\sigma \ge 0$, $\nu \ge 0$, and $\sigma + \nu > 0$. Define $(\phi,\psi,\newomega,\rho,\newC)$ as in \eqref{MF-phi}--\eqref{MF-C}, and assume all of the expectations therein are finite. Then there is a unique strong equilibrium $(\pi^*,c^*)$, and it takes the form given by \eqref{MF-pi*} and \eqref{MF-c*}.
\end{theorem}

\begin{corollary} (Single Stock) \label{co:MFG-singlestock}
Assume that $(\mu, \nu, \sigma)$ is deterministic with $\nu = 0$ and $\mu, \sigma > 0$. Then $\newomega$ defined in \eqref{MF-omega} can be simplified to
\begin{equation*}
\newomega = \frac{\mu^2}{2\sigma^2}(1 - \delta)\left(1 - \frac{\theta}{\theta_{\mathrm{crit}}} \right) \left( \delta - \frac{\theta}{\theta_{\mathrm{crit}}}(\delta -1) \right),
\end{equation*}
where
\begin{equation*}
\theta_{\mathrm{crit}} := \frac{1 + \E[\theta(\delta-1)]}{\E[\delta]},
\end{equation*}
and the optimal investment simplifies to
\begin{equation*}
\pi^* = \left( \delta - \frac{\theta}{\theta_{\mathrm{crit}}}(\delta - 1) \right)\frac{\mu}{\sigma^2}.
\end{equation*}
\end{corollary}

We omit the proof, because it closely parallels the proofs of Theorem \ref{thmN} and \cite[Theorem 3.6]{lacker2017mean}. The main idea, as in the proof of Theorem \ref{thmN}, is to identify the dynamics of the process $\overline{X}_t=\exp\E[\log X_t \, | \, \F^B_t]$, when $X$ is subject to $\F_0$-measurable strategies $(\pi,c)$, with $\pi$ time-independent. The representative agent's optimization problem can then be cast as a (tractable) stochastic control problem over the two-dimensional state process $(X,\overline{X})$.

\section{Discussion of the equilibrium} \label{se:discussion}

We now discuss the interpretation of the equilibria computed in the previous sections and the nature of the dependence on the various model parameters.
First, notice that our result is consistent with Samuelson's \cite{samuelson1975lifetime}, 
in the sense that the investment strategy we obtain is the same as in the model without consumption, derived in the previous work \cite{lacker2017mean}.
More generally, the investment strategy $\pi^*$ does not depend on the relative importance that the agents give to terminal wealth versus consumption, quantified by $\epsilon$ in our model.
With this in mind, we refer to \cite{lacker2017mean} for the discussion on the investment strategy, and we focus the rest of the discussion here on the consumption strategy. 

We further limit the discussion of the equilibrium to the mean field case, for which the equilibrium consumption policy is given by \eqref{MF-c*}, as the equilibrium in the $n$-agent game has essentially the same structure but more complicated formulas. Moreover, we restrict our attention mostly to the single stock case of Corollary \ref{co:MFG-singlestock}, which is again more tractable but already quite rich.

From the expression for $c^*_t$, we can distinguish three regimes of consumption behavior. The optimal consumption is necessarily a monotone function of time $t$, and a quick computation shows that it is increasing when $\newomega < \newC$, decreasing when $\newomega > \newC$ and constant when $\newomega = \newC$.  Recalling the form of the wealth process $X$ in \eqref{wealth}, we see that the expected rate of return of wealth, $\frac{d}{dt}\E[\log X_t \, | \, \F_0]$, is also a monotone function of time, with the opposite monotonicity of the consumption policy. (Note that conditioning on $\F_0$ is equivalent to conditioning on the representative agent's type.)
See Figure \ref{fig:consumption-policies} for some typical consumption policies.

\begin{figure}[h]
    \centering
    \includegraphics[scale =0.4]{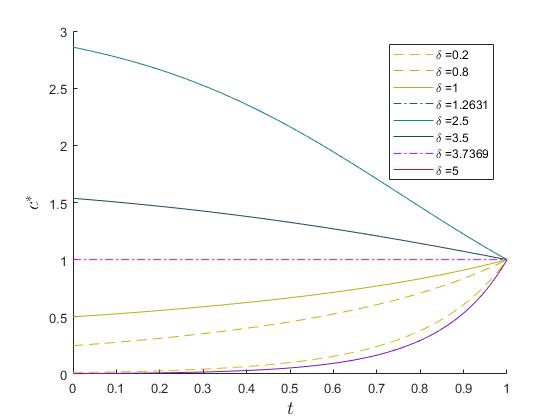}
    \caption{Equilibrium consumption $c^*_t$ versus $t$ for various values of $\delta$. The parameters are $\mu=5$, $\sigma=1$, $\epsilon=1$, $\E[\log(\epsilon^\delta)]=0$, $\E[\theta(\delta-1)] = 0.8$,  $\E[\delta]=3$, $\theta=0.8$, and $T=1$. Note that the final consumption $c^*_T=\newC$ does not depend on $\delta$. }
    \label{fig:consumption-policies}
\end{figure}

Recall that $\epsilon$ captures the relative importance that an agent gives to terminal wealth compared to consumption.
Note that $\newC \to \infty$ as $\epsilon \to 0$, and in particular we have $\newomega < \newC$ for small $\epsilon$. As discussed above, this means the agent aims for a decreasing rate of growth of wealth and an increasing rate of consumption. This is natural, because a small $\epsilon$ indicates the agent's lack of interest in terminal wealth, which drives $X(t)$ toward zero as $t \rightarrow T$ (as there is no bequest in our model). That is, for $\epsilon$ sufficiently small, 
the agent dis-invests after some time in order to consume more. In fact, the agent may even begin dis-investing immediately if $\pi^* < c^*(0)$.
On the contrary, if $\epsilon$ is large, the agent is more concerned with terminal wealth than consumption and will thus decrease her consumption over time. Indeed,  $\newC$ is decreasing in $\epsilon$, so for large $\epsilon$ we have $\newomega > \newC$.

We now turn to the key question of the impact of an agent's competitiveness and risk tolerance on her consumption behavior. To simplify the discussion, let us assume henceforth that no agent has a preference between her utility of wealth or utility of consumption; that is, $\epsilon = 1$ and $\E[\log(\epsilon^{-\delta})] = 0$, which in particular implies $\newC = 1$. 
Note that if $\theta=0$, then we recover the classical Merton solution without competition, with $\newomega= \frac{\mu^2}{2 \sigma^2}\delta(1 - \delta)$ and $\pi^* = \delta\mu/\sigma^2$.
For general $\theta$, we may still rewrite $\newomega$ and $\pi^*$ in an analogous manner as 
\begin{equation*}
\newomega = \frac{\mu^2}{2 \sigma^2}\delta_{\mathrm{eff}}(1 - \delta_{\mathrm{eff}}), \qquad\quad \pi^* = \delta_{\mathrm{eff}}\mu/\sigma^2,
\end{equation*}
where we define the \emph{effective risk tolerance parameter} 
\begin{align*}
\delta_{\mathrm{eff}} := \left(1 - \frac{\theta}{\theta_{\mathrm{crit}}} \right) \delta  + \frac{\theta}{\theta_{\mathrm{crit}}}.
\end{align*}
In other words, in the face of competition, the agent behaves like a Merton investor/consumer but with a different risk tolerance parameter.
We can interpret $\delta_{\mathrm{eff}}$ as a weighted average of the agent's own risk tolerance $\delta$ and the critical log-investor case $\delta_{\mathrm{log}} = 1$, with the weight determined by the agent's competitiveness. Take note, however, that the range of the weight $\theta/\theta_{\mathrm{crit}}$ is $[0,\infty)$ and $\delta_{\mathrm{eff}}$ can be negative, so we should avoid interpreting $\delta_{\mathrm{eff}}$ too literally as a risk tolerance parameter.

Let us investigate in more detail what distinguishes between agents who decrease versus increase their consumption over time, and let us continue to assume that $\epsilon \equiv 1$ (and thus $\newC \equiv 1$) for all agents. As discussed above, an agent increases her consumption over time if and only if $\newomega < \newC=1$. If $8\sigma^2 > \mu^2$, then we always have $\newomega < 1$, because $\delta_{\mathrm{eff}}(1-\delta_{\mathrm{eff}}) \le 1/4$ for any $\delta_{\mathrm{eff}}$. So assume instead that $8\sigma^2 < \mu^2$. Then, because $\newomega$ is a quadratic function of $\delta$ (if all other parameters are held fixed), we may find $\delta^*_{\pm}$ such that
\begin{equation*}
\newomega > 1 \ \ \Longleftrightarrow \ \  \delta \in (\delta_-^*, \delta_+^*).
\end{equation*}
That is, the agent decreases consumption over time if $\delta_-^* < \delta < \delta_+^*$, increases consumption over time if $\delta \notin (\delta_-^*, \delta_+^*)$, and consumes at a constant rate if $\delta \in \{\delta_-^*, \delta_+^*\}$.
Precisely, these two values are
\begin{equation*}
\delta_{\pm}^* := 1 + \frac{1}{2} \left( \frac{1}{\theta / \overline{\theta}_{\mathrm{crit}}- 1} \pm \frac{\sqrt{1 - 8 \sigma^2 / \mu^2}}{\left| \theta/ \overline{\theta}_{\mathrm{crit}} -1\right|} \right).
\end{equation*}
(If $\theta=\theta_{\mathrm{crit}}$, then $\newomega=0$, so let us assume $\theta \neq \theta_{\mathrm{crit}}$.) 
This explains the non-monotonicity in $\delta$ of the equilibrium consumption strategy, as well as the wave-like shape of the curve of $c^*_t$ versus $\delta$ and $\theta$ pictured in Figure \ref{fig:opt_c_vs_delta_theta}. 
In the classical Merton problem with no competition, recovered by setting $\theta=0$, the endpoints become $\delta_{\pm}^* = \tfrac12(1 \pm \sqrt{1 - 8 \sigma^2 / \mu^2})$, both of which are less than $1$; in this case only risk averse ($\delta < 1$) agents decrease their consumption over time. In contrast, in the competitive case $\theta > 0$, the interval $(\delta_-, \delta_+)$ may lie above or below $1$ depending on the sign of $1 - \theta /\overline{\theta}_{\mathrm{crit}}$.

\begin{figure}[h]
    \centering
    \includegraphics[scale =0.7]{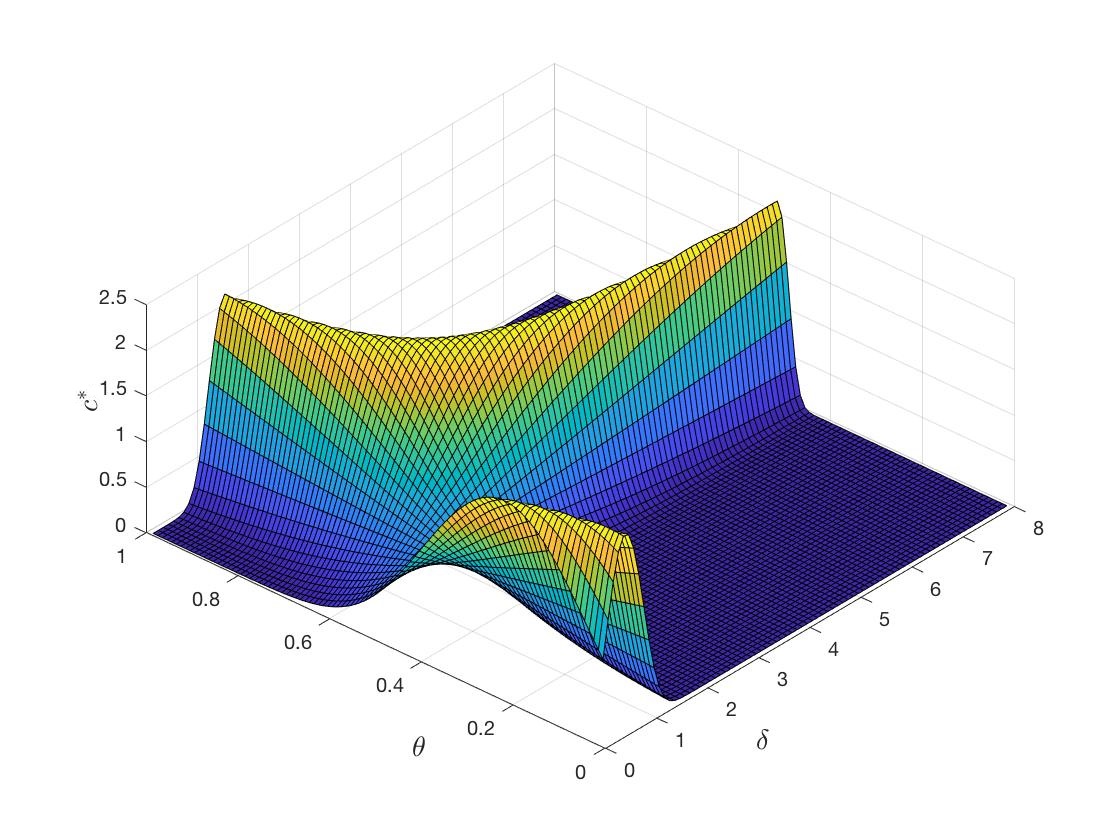}
    \caption{Optimal consumption versus $(\delta,\theta)$ at time $T/2$. The parameters are $\mu=5$, $\sigma=1$, $\epsilon=1$, $\E[\log\epsilon]=0$, $\E[\theta(\delta-1)]=1.6$, $\E[\delta]=5$, and $T=1$.}
    \label{fig:opt_c_vs_delta_theta}
\end{figure}

On the one hand, suppose the agent is less competitive than the critical value, or $\theta < \overline{\theta}_{\mathrm{crit}}$. Then $\delta_-^* < \delta_+^* < 1$, and we have seen that the agent will decrease her rate of consumption over time only if her risk tolerance lies within the range $(\delta_-^*, \delta_+^*)$. 
As we would expect, a relatively uncompetitive agent behaves similarly to a Merton investor in this respect.
On the other hand, if the agent is more competitive than the critical value, or $\theta > \overline{\theta}_{\mathrm{crit}}$, then $1 < \delta_-^* < \delta_+^*$. This means that even highly risk tolerant agents may decrease consumption over time.
Noting that $\delta_{\mathrm{eff}}$ decreases with $\theta$, one interpretation is as follows: Increasing $\theta$ exposes an agent to relative performance pressures, which is itself a source of risk, and to offset this additional risk the agent behaves like a Merton investor with a smaller risk tolerance parameter.

Note that a highly competitive agent, with $\theta > \theta_{\mathrm{crit}}$, behaves in a sense \emph{opposite} to how they would if $\theta=0$. Indeed, when $\theta > \overline{\theta}_{\mathrm{crit}}$, the effective risk tolerance $\delta_{\mathrm{eff}}$ is less than $1$ if $\delta > 1$  and greater than $1$ if $\delta < 1$. The agent effectively switches to the other side of the critical risk tolerance $\delta_{\mathrm{log}}=1$.

We have seen by now how, with other parameters held fixed, the consumption policy may depend non-monotonically on the risk tolerance $\delta$, with an intermediate range of risk tolerance parameters $(\delta_-^*,\delta_+^*)$ in which agents decrease consumption over time. Similarly, with other parameters held fixed, consumption can exhibit the same non-monotonicities as a function of $\theta$, with an intermediate range $(\theta_-^*,\theta_+^*)$ in which agents decrease consumption over time. See Figure \ref{fig:region} for a depiction of the range of $(\delta,\theta)$ parameters leading agents to decrease versus increase consumption over time.

\begin{figure}[h]
    \centering
    \includegraphics[scale =0.4]{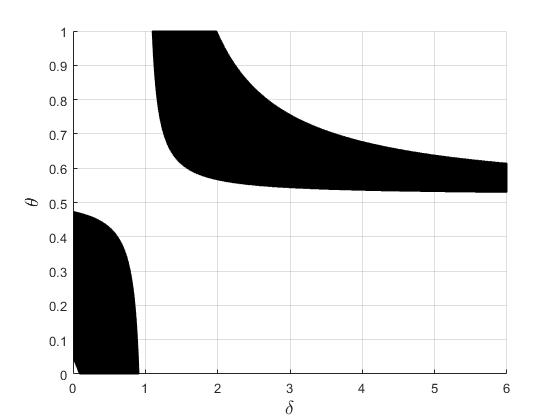}
    \caption{Consumption regime versus $(\delta,\theta)$. An agent with $(\delta,\theta)$ lying inside (resp.\ outside) the shaded region decreases (resp.\ increases) consumption rate over time. Agents on the boundary consume at a constant rate. Note there is a small unshaded wedge near the origin. The parameters are $\mu=5$, $\sigma = 1$, $\epsilon=1$, $\E[\log \epsilon]=0$, $\E[\theta(\delta-1)] = 1.6$, $\E[\delta]=5$. Here $\theta_{\mathrm{crit}}=0.52$.}
    \label{fig:region}
\end{figure}

\bibliographystyle{plain}
\bibliography{biblio}

\end{document}